\documentclass[aip,jmp,12pt]{revtex4}

\draft % marks overfull lines with a black rule on the right

\usepackage[dvips]{graphicx}

\usepackage{bm}% bold math
\usepackage{verbatim}       % Defines \begin{comment}, \end{comment}

\usepackage{amsthm}
\usepackage{amssymb}

\usepackage{psfrag}
\usepackage{bm}
\usepackage{amstext}
\usepackage{psfrag}
\usepackage{amsmath}
\usepackage{amsfonts}
\usepackage{mathtools}
\usepackage{units}
\usepackage{mathrsfs}
\newcommand{\bd}{\begin{definition}}                %inizia definizione
\newcommand{\ed}{\end{definition}}                  %fine definizione
\newcommand{\bc}{\begin{corollary}}                 %inizia corollario
\newcommand{\ec}{\end{corollary}}                   %fine corollario
\newcommand{\bl}{\begin{lemma}}                     %inizia lemma
\newcommand{\el}{\end{lemma}}                       %fine lemma
\newcommand{\bp}{\begin{proposition}}            %inizia proposizione
\newcommand{\ep}{\end{proposition}}                %fine proposizione
\newcommand{\bere}{\begin{remark}}                  %inizia osservazione
\newcommand{\ere}{\end{remark}}                     %fine oservazione

\newcommand{\bt}{\begin{theorem}}
\newcommand{\et}{\end{theorem}}

\newcommand{\be}{\begin{equation}}
\newcommand{\ee}{\end{equation}}

\newcommand{\bit}{\begin{itemize}}
\newcommand{\eit}{\end{itemize}}
\newtheorem{theorem}{Theorem}[section]
\newtheorem{corollary}[theorem]{Corollary}
\newtheorem{lemma}[theorem]{Lemma}
\newtheorem{proposition}[theorem]{Proposition}
\theoremstyle{definition}
\newtheorem{definition}[theorem]{Definition}
\theoremstyle{remark}
\newtheorem{remark}[theorem]{Remark}
\newtheorem{example}[theorem]{Example}

\begin{document}

\title{Causally simple inextendible spacetimes are hole-free}

%\author{E. Minguzzi\thanks{
%Dipartimento di Matematica Applicata ``G. Sansone'', Universit\`a
%degli Studi di Firenze, Via S. Marta 3,  I-50139 Firenze, Italy.
%E-mail: ettore.minguzzi@unifi.it} }

\author{E. Minguzzi}
\email[]{ettore.minguzzi@unifi.it}
%\homepage[]{Your web page}
%\thanks{}
%\altaffiliation{}
\affiliation{Dipartimento di Matematica Applicata ``G. Sansone'',
Universit\`a degli Studi di Firenze, Via S. Marta 3, I-50139
Firenze, Italy.}

%\date{}

%\maketitle

\begin{abstract}
It is shown that causally simple inextendible spacetimes are
hole-free, thus confirming the expectation that causal simplicity
removes holes from spacetime. This result is optimal in the sense
that causal simplicity cannot be weakened to causal continuity.
Physically, it means that if there is some partial Cauchy
hypersurface which, for some reason, does not fully develop its
influence, then there is some  discontinuity in the causal relation.
\end{abstract}


\maketitle

\section{Introduction}

With the general theory of relativity we learned to represent events
as points of a connected Hausdorff time oriented Lorentzian
manifold. We call these manifolds {\em spacetimes}. However, it is
easy to recognize that this category of manifolds is too large.
Indeed, removing points from one spacetime leaves us with another
spacetime, thus, taking the mathematical model seriously, one should
conclude that spacetime might suddenly end without giving any
warning to the observers living in it. It seems that to accept this
type of drastic behavior would mean to dismiss the possibility of
making predictions. Therefore, there is a certain agreement that the
above notion of spacetime should be narrowed with the imposition of
additional conditions. One  such condition is {\em inextendibility},
namely, the Lorentzian manifold should not be a subset of a larger
one. Physically, this condition assures that spacetime will not end
for some observer (i.e. timelike worldline) if there is the
possibility of hosting the observer by enlarging the manifold.

Another, less known, but soundly motivated condition is that of {\em
hole-freeness}. An inextendible spacetime is hole-free if the domain
of dependence of any partial Cauchy hypersurface cannot be extended
by imbedding it into a different spacetime. In essence it states
that any Cauchy development on spacetime must be maximal (in the
sense of partial orders) among the Cauchy developments of that
hypersurface (in possibly different spacetimes). Philosophically
speaking, it means that spacetime evolves as much as possible,
without `forgetting' any event that it can possibly determine. This
idea was introduced by Geroch in order to identify some pathological
spacetimes which could not be physically dismissed on other grounds.
According to Geroch, general relativity must be complemented with
the condition of hole-freeness.
%An hole-free spacetime can be recovered mindlessly
%patching together its globally hyperbolic subsets.

The inextendibility condition allows us to discard spacetimes which
are proper subsets of physically reasonable spacetimes. The
hole-free condition has a similar effect but can be formulated
independently of the  inextendibility condition (see the examples of
section \ref{exa}). Since both should be imposed by fiat, they
should most often appear in conjunction.

Causal simplicity is another condition which, intuitively, forbids
the removal of points from spacetime. Indeed, the levels of the
causal ladder \cite{hawking73,minguzzi06c} above stable causality,
such as causal continuity, causal simplicity and global
hyperbolicity appear as increasingly refined conditions which
decrease the influence of points at infinity on spacetime, and
hence, in particular, of those points at infinity obtained by
removing interior spacetime points. Causal simplicity seems to work
perfectly fine on known examples although, contrary to
hole-freeness, it cannot be imposed by fiat and could fail even in
physically meaningful spacetimes (e.g. plane-fronted waves).

In this work I wish to show that, as examples suggest, causal
simplicity together with inextendibility implies hole-freeness. I
will also comment on previous variants on the definition of
hole-free spacetime.

%I refer the reader to \cite{minguzzi06c} for most of the conventions
%used in this work. In particular,
I denote with $(M,g)$ a spacetime (connected, Hausdorff,
time-oriented Lorentzian manifold without boundary), of arbitrary
dimension $n\geq 2$ having a metric $g$ of signature
$(-,+,\dots,+)$.  Unless otherwise specified, all the causal curves
that we shall consider are future directed, thus a past inextendible
curve ends at its future endpoint. Sometimes by {\em geodesic} we
shall actually mean {\em pregeodesic}. The subset symbol $\subset$
is reflexive, $X \subset X$.

\section{Hole-free spacetimes}

The future domain of dependence (or future Cauchy development)
$D^+(S)$ of a set $S$ is made by all the points $p\in M$ for which
any inextendible ($C^1$) causal curve through $p$ intersects $S$.
The set $\tilde{D}^+(S)$ is defined analogously, where {\em causal}
is replaced by {\em timelike}.\cite{hawking73} If $S$ is
topologically closed as a subset of $M$ then
$\overline{D^+(S)}=\tilde{D}^+(S)$. A dual definition of past domain
of dependence $D^-(S)$ can be given. The (total) domain of
dependence is $D(S)=D^+(S)\cup D^{-}(S)$.

An acausal edgeless (and hence topologically closed) set is a {\em
partial Cauchy hypersurface}. For the definition of edge of an
achronal set see\cite{hawking73}. Essentially, the edge is the
boundary of the achronal set $S$ with respect the the topology
induced on an maximal achronal set $A$ containing $S$.

Let $S$ be acausal. It is well known that through each point of
$H^+(S)$ passes a past inextendible lightlike geodesic generator
whose intersection with $\bar{S}$, whenever it exists, belongs to
$\textrm{edge}(S)$ (and dually in the past case).

%%Ciao, I am Etsuko, I only write commented senteces.

As a consequence, if $S$ is acausal and does not contain its edge
(e.g. if $S$ is a partial Cauchy hypersurface) then $D(S)$ is open
and non-empty (as it contains $S$), and furthermore $D(S)\cap
H(S)=\emptyset$. We remark that we cannot weaken the acausality  of
$S$ to achronality if we want to assure that $D(S)$ is open.

A spacetime is globally hyperbolic iff we can find a partial Cauchy
hypersurface such that $D(S)=M$. It is causally simple if it is
causal and the sets $J^{\pm}(x)$ are closed for every $x\in M$ (or
equivalently, the causal relation $J^{+}\subset M\times M$ is closed
\cite{minguzzi06c}). It is causally continuous if it is weakly
distinguishing and reflective.\cite{minguzzi07e} Global
hyperbolicity implies causal simplicity which implies casual
continuity.

Any isometry is injective but not necessarily surjective. A
spacetime $(M,g)$ is inextendible if it is not contained in a larger
spacetime $(M',g')$, namely if we cannot find an isometry $\psi:
M\to M'$ such that $\psi(M)\ne M'$.

\begin{definition} \label{gvw}
A spacetime $(M,g)$ has a {\em future Cauchy hole} (or simply a
future hole) if there is a partial Cauchy hypersurface $S$ and an
isometry $\varphi: \tilde{D}(S) \to N$, on a spacetime $(N,\sigma)$,
such that $\varphi(S)$ is acausal and $\varphi(H^+(S))\cap
D^{+}(\varphi(S))\ne \emptyset$. The definition of {\em past Cauchy
hole}  is given dually. A spacetime is {\em Cauchy holed} if it has
a future or a past hole. A spacetime is (future/past) hole-free if
it has no (future/past) hole.
\end{definition}

Intuitively, $S$ generates a future hole if its Cauchy horizon
$H^{+}(S)$, while outside the domain of influence of $S$ in $M$,
enters into the domain of influence in an alternative spacetime,
that is, prediction can be extended passing through the former
horizon.

According to this definition the region $t<0$ of Minkowski spacetime
is not holed (compare with previous definitions, see remark
\ref{rem}).

\begin{proposition} \label{pyn}
Let $S$ be a closed achronal set and let  $\varphi: \tilde{D}(S) \to
N$ be an isometry, then $\varphi(D(S)) \subset {D}(\varphi(S))$ and
$\varphi(\tilde{D}(S)) \subset \tilde{D}(\varphi(S))$.
\end{proposition}

\begin{proof}
Let $q'\in \varphi(D^+(S))\backslash \varphi(S)$, $q'=\varphi(q)$,
$q\in D^{+}(S)\backslash S$. There is a neighborhood of ${0}\in
TM_q$, such that the exponential map of the past non-spacelike
vectors in that neighborhood is contained in $D^+(S)$. By the
isometry, the same is true for $q'$ with respect to
$\varphi(D^+(S))$, in particular there is a neighborhood $U_{q'}$ of
$q'$ such that $U_{q'}\cap J^{-}(q')\subset \varphi(D^+(S))$.
Suppose by contradiction that $q'\notin D^+(\varphi(S))$, then there
is a past inextendible causal curve \mbox{$\gamma':(0,1]\to N$}
ending at $q'$ and not intersecting $\varphi(S)$. A last segment of
$\gamma'$ must be contained in $\varphi(D^+(S))$. The curve
$\gamma:=\varphi^{-1}\circ \gamma'$ is causal and ends at $q$ thus
there is a segment of $\gamma$ which is ending at $q$ and entirely
contained in $D^+(S)$. In fact there is a maximal segment of this
type which provides a past inextendible causal curve ending at $q$
not intersecting $S$, a contradiction. The proof in the past and
timelike cases are analogous.
\end{proof}

%\begin{proof} Let $q'\in \varphi(D^+(S))\backslash \varphi(S)$,
%$q'=\varphi(q)$, $q\in D^{+}(S)\backslash S$, and suppose by
%contradiction that $q'\notin D^+(\varphi(S))$ then there is a past
%inextendible causal curve $\gamma':(0,1]\to N$ ending at $q'$ and
%not intersecting $\varphi(S)$. Since $\varphi$ is an isometry, it is
%a local diffeomorphism and hence it is open. Since $\varphi$ is open
%its image includes some neighborhood of $q'$.
%\end{proof}

%Thus the hole-free condition imposes the inclusion $D(\varphi(S))
%\subset \varphi(D(S))$.

If there is a hole the inclusion is strict.

\begin{proposition} \label{muj}
Suppose that $(M,g)$ is (future/past) holed and let $\varphi$ be as
in definition \ref{gvw}, then $\varphi(D(S))\subsetneq
D(\varphi(S))$ (resp. $\varphi(D^{\pm}(S))\subsetneq
D^{\pm}(\varphi(S))$).
\end{proposition}

\begin{proof}
Indeed, by definition of hole $D(\varphi(S))$ contains some points
of $\varphi(H(S))$ while $\varphi(D(S))$ does not contain any of
them since $\varphi: \overline{D(S)} \to N$ is injective.
\end{proof}

\begin{remark}
Let $\varphi: \tilde{D}(S) \to N$ be an isometry  with $S$ partial
Cauchy hypersurface, and $\varphi(S)$ acausal set. Since $\varphi$
is an isometry, it is a local diffeomorphism and hence it is open.
The set $D(S)$ is open thus $\varphi(D(S))$ is open. In particular,
as $S\subset D(S)$, $\varphi(S)$ is contained in the open set
$\varphi(D(S))$ and no point of $\varphi(S)$ can belong to
$\textrm{edge}(\varphi(S))$ (we recall \cite[Prop. 6.5.2]{hawking73}
that this set coincides with $\textrm{edge}(H(S))$). By hypothesis
the hypersurface $\varphi(S)$ is acausal and furthermore, it does
not contain its edge, thus $D(\varphi(S))$ is open and
$H(\varphi(S))\cap D(\varphi(S))=\emptyset$.
\end{remark}

%Clearly, $\varphi(D(S)) \subset {D}(\varphi(S))$ is guaranteed by
%the isometry (which also implies $\varphi(\tilde{D}(S)) \subset \tilde{D}(\varphi(S))$).

\begin{proposition} \label{nrx}
Suppose that $(M,g)$ is  inextendible, $S$ is a partial Cauchy
hypersurface and let $\varphi: \tilde{D}(S) \to N$ be an isometry on
a spacetime $(N,\sigma)$, such that $\varphi(S)$ is acausal. Then
\[[D^{+}(\varphi(S))\cap \partial \varphi(D^+(S))]\backslash
\varphi(S)= D^{+}(\varphi(S))\cap\overline{\varphi(H^+(S))},\] and
if $M$ is future holed then $S$, $(N,\sigma)$ and $\varphi$ can be
chosen in such a way that this set is non-empty.
\end{proposition}

\begin{proof}
The definition of future hole implies that the right-hand side is
non-empty for some choice of $S$, $(N,\sigma)$ and $\varphi$.

The set $\varphi(H^+(S))$ has empty intersection with
$\varphi(D^{+}(S))$ because $\varphi: \overline{D^+(S)}\to N$ is
injective. However, by continuity, since $H^{+}(S)\subset
\overline{D^{+}(S)}$ we have $\varphi(H^{+}(S))\subset
\overline{\varphi(D^{+}(S))}$, thus $\varphi(H^{+}(S))\subset
\overline{\varphi(D^{+}(S))} \backslash \varphi(D^{+}(S))\subset
\partial \varphi(D^+(S))$. Using the closure of the last set we get
$\overline{\varphi(H^{+}(S))}\subset
\partial \varphi(D^+(S))$ which proves $\supset$.

For the inclusion $\subset$, suppose that $r\in
D^{+}(\varphi(S))\cap
\partial \varphi(D^+(S))]\backslash \varphi(S)$ and that $r\notin
\overline{\varphi(H^+(S))}$. We can find an open neighborhood $U \ni
r$, $U\subset I^{+}(\varphi(S))\cap D^{+}(\varphi(S))$, such that
$U\cap \varphi(H^+(S))=\emptyset$. Gluing $U$ with $M$ by
identifying $\varphi^{-1}(U)\subset D(S)\subset M$ with
$\varphi(\varphi^{-1}(U))$ through $\varphi$ we get an extension of
$M$, a contradiction.

\end{proof}

\begin{proposition} \label{nja}
An inextendible spacetime $(M,g)$ is {\em hole-free} if and only if
for every partial Cauchy hypersurface
 $S$, and isometry $\varphi: \tilde{D}(S) \to N$, on a spacetime
$(N,\sigma)$, such that $\varphi(S)$ is acausal, we have
$\varphi(D(S))=D(\varphi(S))$.
\end{proposition}

\begin{proof}
$\Leftarrow$. This direction is an easy consequence of Prop.
\ref{muj}.

$\Rightarrow$. If $(M,g)$ has no hole then  for every partial Cauchy
hypersurface $S$ and isometry $\varphi: \tilde{D}(S) \to N$, on a
spacetime $(N,\sigma)$, such that $\varphi(S)$ is acausal we have
$\varphi(H^\pm(S))\cap D^{\pm}(\varphi(S))= \emptyset$. By Prop.
\ref{nrx}  we have $D(\varphi(S))\cap \partial
\varphi(D(S))=\emptyset$. Let us show that it implies
$D(\varphi(S))\subset \varphi(D(S))$. Indeed, suppose by
contradiction that there is $q'\in D(\varphi(S))\backslash
\varphi(D(S))$.
%Since $D(\varphi(S))$ is open we can actually assume
%$q'\in D(\varphi(S))\backslash \overline{\varphi(D(S))}$.
We shall assume $q'\in D^+(\varphi(S))$, the case $q'\in
D^-(\varphi(S))$ being analogous. There is a  future directed
timelike curve $\gamma':[0,1]\to D(\varphi(S))$ which starts from
some point $p'\in \varphi(S)$ and reaches $q'$. Thus $\gamma'$
starts from the open set $\varphi(D(S))$ and reaches a point $q'$
outside it. Thus there is a first point $r=\gamma(b)\in
D^+(\varphi(S))\backslash \varphi(D(S))$, $b\in (0,1]$, hence
belonging to $[\partial \varphi(D^+(S))]\backslash \varphi(S)$.
Observe that $r \in D^{+}(\varphi(S))$ since it belongs to
$J^{-}(q')\cap I^{+}(\varphi(S))$. This fact contradicts
$D(\varphi(S))\cap \partial \varphi(D(S))=\emptyset$ and proves the
inclusion $D(\varphi(S))\subset \varphi(D(S))$. Using Prop.
\ref{pyn} we obtain the equality $\varphi(D(S))=D(\varphi(S))$.
\end{proof}

\begin{remark}[On definitions] \label{rem}
A first definition of hole-free spacetime can be found in Geroch,
\cite{geroch77} who should be credited with the introduction of the
concept. He gives a definition similar to the characterizing
property of Prop. \ref{nja}. In his definition there is no edgeless
condition on $S$ and, furthermore, he uses $\tilde{D}(S)$ in place
of $D(S)$ in the last equality, which causes
 some undesired consequences related to the non openness of
$\tilde{D}(S)$ (see Manchak \cite{manchak09}). Manchak mentions a
possible correction recently suggested by Geroch himself, which
makes the definition almost coincident with the property of Prop.
\ref{nja} but for the fact that: (i) he uses the achronal condition
on $S$, and expresses it using $\tilde{D}(S)$ and, (ii) he misses
the inextendibility condition with the unfortunate consequence that
the region $t<0$ of Minkowski spacetime becomes holed according to
his definition, (iii) $\varphi$ is defined over $D(S)$ instead that
over $\overline{D(S)}$. We choose to define $\varphi$ over
$\overline{D(S)}$ because this choice allows us to skip some
technical issues connected with the extension of the isometry from
the open set $D(S)$ to its closure (by the way, Krasnikov requires
the isometry to be defined in a neighborhood of $\overline{D(S)}$).

%Also, we want   $H(S)$ to be included in $N$, that is, that $N$ does
%not miss any point that should be there by continuity

Another definition can be found in Clarke \cite{clarke76} where he
credits J. Earman and N. Woodhouse (Earman refers to the hole-free
condition as the {\em determinism maximal} property \cite[Sect.
3.8]{earman95}). However, Clarke uses spacelike hypersurfaces
(without edge) in place of partial Cauchy hypersurfaces. Manchak
points out that Clarke's definition runs into the same problems of
Geroch's. Indeed, Clarke's uses $\tilde{D}(S)$ in his definition
\cite[p.119]{clarke93}. Though he uses edgeless hypersurfaces for
$S$, this fact does not guarantee that $\tilde{D}(S)$ is open.
Krasnikov has also shown that if the free-hole condition is taken as
in Geroch's first paper then even Minkowski's spacetime is
holed.\cite{krasnikov09}

In a way Clarke's definition is more general, as it does not require
any causality condition, not even the acausality of $S$.
Nevertheless, a treatment of this case would require a considerable
generalization of the terminology, together with unnecessary
complications, especially whenever $N$ can be chosen totally
vicious. For instance, according to the usual definition of horizon,
if $N$ is a totally vicious spacetime then $H(\varphi(S))$ is empty
although it could be $D(\varphi(S))\ne N$. Of course, this is not
what one usually expect from a domain of dependence,\cite{hawking73}
as we use to regard the horizon as the boundary of the  Cauchy
development.

%Thus, for simplicity we will use the above definition,  also in
%consideration of the fact that in presence of causality violations
%one is  more concerned of those pathologies than of the presence of
%holes. This point of view is consistent with what we learned from
%the causal ladder of spacetime. There the removal of holes starts
%naturally only at the strongest levels of the ladder.

%The definition of hole by Krasnikov is in a way more restrictive, as
%he demands that a whole neighborhood of

\end{remark}

Clarke \cite[Prop. 6.5.1]{clarke93} claims that any inextendible
globally hyperbolic spacetime is hole-free. His proof goes as
follows. He first  considers a globally hyperbolic holed spacetime
from which he builds a potentially branched extension of $M$. Using
global hyperbolicity he argues that the extension is acceptable as
it respects Hausdorffness (thus in the end it is not branched)(this
proof is not  clear to this author, indeed his set $N'$ is closed
and the open sets $\pi(U)$ there defined are not homeomorphic to
subsets of $\mathbb{R}^n$).

%Furthermore the topology on the quotient is richer than claimed.

%
%
%His considers a partial Cauchy surface $S$ and supposes that
%$\tilde{D}(S)$ can be isometrically embedded in a spacetime $N$,
%$\varphi: \tilde{D}(S) \to N$, in such a way that
%$\varphi(D^+(S))\subsetneq D^+(\varphi(S))$.
%
%
%which
%
%
%
%However, an inspection of his proof shows that his inextendibility
%condition is much stronger than the usual one. Indeed, he does not
%impose the Hausdorff condition on the manifold, so that, by assuming
%that the manifold is globally hyperbolic and contains a hole he is
%able to construct a {\em branching} extension (the branching takes
%place at the horizon created by the hole). Clarke had been
%considering branching spacetimes in his previous works
%\cite{clarke76}, but I did not find explicit mention in his book
%\cite{clarke93}, a rather unfortunate circumstance as one can easily
%misread his result. Furthermore, any spacetime is extendible in a
%branching sense so that Clarke's theorem seems to be vacuous (the
%problems with his argument become clear if we let $M$ be the
%rectangle $(-1/4,1/4)\times (-1,1)$ of 1+1 Minkowski spacetime with
%coordinates $(t,x)$ and we let $S$ be the portion of the parabola
%$t=-x^2
%
%
%the globally hyperbolic
%
%
%
%we think of example \ref{miv} of the next section). Since we impose
%the Hausdorff condition on spacetimes, we do not consider this type
%of extensions as admissible.

%Nevertheless, we are able to prove the same result with the usual
%notion of inextendibility . In fact we obtain the

We are able to prove the stronger result that inextendible causally
simple spacetimes are hole-free. This theorem justifies the
intuitive claim that ``causal simplicity removes holes from
spacetime''. Physically it means that any Cauchy hole on spacetime
causes some discontinuity in the causal relation.

\begin{theorem} \label{nxp}
Every inextendible and causally simple spacetime is hole-free.
\end{theorem}

\begin{proof}
Let us suppose, by contradiction, that $(M,g)$ is inextendible,
causally simple, and that it has a future hole. Thus there is a
partial Cauchy hypersurface $S$, a spacetime $(N,\sigma)$, and an
isometry $\varphi: \overline{D(S)}\to N$ such that $\varphi(S)$ is
acausal, and
\[
B:=D^{+}(\varphi(S))\cap\overline{\varphi(H^+(S))}=[D^{+}(\varphi(S))\cap
\partial \varphi(D^+(S))]\backslash \varphi(S)
\]
is non-empty (Prop. \ref{nrx}).

We remark that starting from any point of $B$ and moving backwards
along a timelike curve one enters immediately into
$\varphi(D^+(S))$.  Let $r\in D^{+}(\varphi(S))\cap
{\varphi(H^+(S))}$ and let $x=\varphi^{-1}(r)$, $x\in H^{+}(S)$.
There is a past inextendible geodesic $\eta':(0,1]\to N$ ending at
$r$ such that $\eta:=\varphi^{-1}\circ \eta'\vert_{(a,1]}$ is, for
some $a>0$, a past inextendible lightlike generator of $H^{+}(S)$
ending at $x$ (see Fig. \ref{fhole}). Observe that $\eta$ does not
intersect $S$, since $S$ is edgeless, nevertheless, $\eta'$
intersects $\varphi(S)$ as $r\in D^{+}(\varphi(S))$. Thus following
$\eta'$ in the backward direction we should reach a last point $p\in
B$. Now, let $C$ be a convex neighborhood of $p$ such that $C\subset
W\subset D^{+}(\varphi(S))\backslash \varphi(S)$ where $W$ is
causally convex in the globally hyperbolic spacetime
$D(\varphi(S))$. Any point $q\in J^{-}_{C}(p)\backslash \{p\}$
cannot belong to $\overline{\varphi(H^+(S))}$, for this is true for
the points staying in $\eta'\cap J^{-}_{C}(p)$ and if $q$ is not of
this type then it can be connected with a timelike curve $\sigma'$
to $r$ (as it can be connected to $r$ through a piecewise geodesic
causal curve passing through $p$ with a corner at $p$), and hence,
if it were $q\in \overline{\varphi(H^+(S))}$  we could find $q'\in
\varphi(H^+(S))\cap C$ and a timelike curve $\beta'$ which connects
$q'$ with $r$. If $\beta'$ is entirely in the image of $\varphi$
then the existence of $\beta:= \varphi^{-1}\circ \beta'$ contradicts
the achronality of $H^{+}(S)$. If instead, $\beta'$ is not entirely
in the image of $\varphi$ the curve $\varphi^{-1}\circ \beta'$
provides a past inextendible causal curve ending at $x$ and not
intersecting $S$, again a contradiction. In summary, we have shown
that there is a point $p\in B$ and a convex neighborhood $C\ni p$,
$C\subset W\subset D^{+}(\varphi(S))\backslash \varphi(S)$ such that
$[J^{-}_{C}(p)\backslash \{p\}]\cap
\overline{\varphi(H^{+}(S))}=\emptyset$ and hence
$J^{-}_{C}(p)\backslash \{p\} \subset \varphi(D^{+}(S))$. In
particular, $p\notin \varphi(H^+(S))$, for otherwise the image of
the generator of $H^{+}(S)$ passing through $\varphi^{-1}(p)$ would
provide points belonging to $[J^{-}_{C}(p)\backslash \{p\}]\cap
\varphi(H^{+}(S))$.

\begin{figure}[ht]
\centering
%\psfrag{s}{{\footnotesize $S$}} \psfrag{g}{{\footnotesize
%{\large $\varphi$}}}  \psfrag{c}{$C$} \psfrag{b}{$\beta'$}
%\psfrag{p}{$p$} \psfrag{d}{{\small $\varphi(D^+(S))$}}
%\psfrag{f}{{\footnotesize $ \varphi(H^+(S))$}}
%\psfrag{h}{{\footnotesize $ H^+(S)$}} \psfrag{j}{{\footnotesize
%$H^+(\varphi(S))$}} \psfrag{z}{{\footnotesize $ \varphi(S)$}}
%\psfrag{r}{$r$} \psfrag{n}{$N$} \psfrag{m}{$M$}
%\psfrag{e}{{\footnotesize $ D^{+}(\varphi(S))\cap
%\partial \varphi(D^+(S))]\backslash [\varphi(S)\cup \overline{\varphi(H^+(S))}]$}}
\includegraphics[width=11cm]{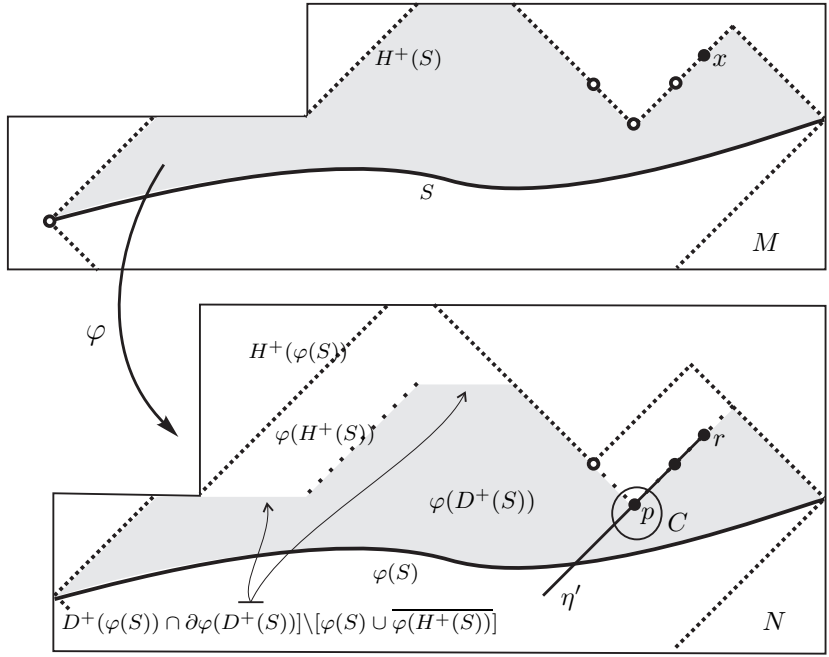}
\caption{An example of future holed spacetime obtained by removing
some sets from Minkowski 1+1 spacetime. This example illustrates the
definition of hole and the proof of theorem \ref{nxp}. The circled
empty points do not belong to the manifold. If the set indicated by
the arrows is non-empty then $M$ can be easily extended.}
\label{fhole}
\end{figure}

Let $p_n \to p$, $p_n \in \varphi(H^+(S))$, and let $\gamma'_n$ be
the past inextendible lightlike geodesic ending at $p_n$ such that
the pullback $\varphi^{-1}\circ \gamma_n'$ (with suitably restricted
domain) is the generator of $H^{+}(S)$ ending at
$\varphi^{-1}(p_n)$. The curves $\gamma_n'$ converge to a past
inextendible curve ending at $p$. Since there is an open set
$V\subset C\cap \varphi(D^+(S))$ of $J^{-}_{C}(p)\backslash \{p\}$,
for sufficiently large $n$ we must have that some $\gamma_n'$
intersects $V$. Let $N$ be such that there is  $b_N\in \gamma_N'\cap
V$. Following $\gamma'_N$ from $p_N$ in the past direction we must
leave $\overline{\varphi(H^+(S))}$ before $b_N$ at some last point
$c$. Point $c$  cannot belong to $\varphi(H^+(S))$ (for otherwise it
would not be the last point), thus $c\in
\overline{\varphi(H^+(S))}\backslash \varphi(H^+(S))$. Then we get a
contradiction with causal simplicity because $(\varphi^{-1}(b_N),
\varphi^{-1}(p_N))\in \overline{J^{+}}\backslash J^{+}$. In order to
show that $(\varphi^{-1}(b_N), \varphi^{-1}(p_N))\in
\overline{J^{+}}$ it is sufficient to consider a sequence of points
$b_N^k\to b_N$, $b_N^k\ll b_N$. We can find timelike curves
$\alpha_k$ connecting $b_N^k$ to $p_N$, thus  entirely contained in
the image of $\varphi$. Their pullback shows that
$\varphi^{-1}(b_N^k) \ll \varphi^{-1}(p_N)$ where
$\varphi^{-1}(b_N^k) \to \varphi^{-1}(b_N)$.

%Let $W=B\backslash \varphi(H^+(S))\subset \textrm{edge}(
%\varphi(H^+(S))\cap D^{+}(\varphi(S))$ so that $p\in W$. We can find
%$y,z\in C$, $y\ll p\ll z$, such that there is a timelike curve
%contained in $C$ joining $y$ to $z$ and not intersecting
%$\varphi(H^+(S))$. This curve must intersect $E^{+}_{C}(p)$ at some
%point $x$.  This point cannot be $p$ because $p\notin
%\varphi(H^+(S))$. Let us consider a point $x'\in C$, $x'<_C p$,
%belonging to the (unique) lightlike geodesic of $C$ connecting $p$
%to $x$. Since by the special choice of $p$, $x'$ cannot belong to
%$B$ it must belong to $\varphi(D^{+}(S))$.

\end{proof}

\begin{remark}
In the previous proof it could be tempting to claim that the pair
$(\varphi^{-1}(e),x)$, where $e$ is the intersection  of $\eta'$
with $S$, belongs to $\overline{J^{+}}\backslash J^{+}$ and provides
a pair of points giving the desired contradiction. Unfortunately,
this simplified strategy does not work because there could be
conjugate points in $\eta'$ between $p$ and $r$. Their presence
would actually imply $r\in I^{+}(e)$.
\end{remark}

Every extendible spacetime is causally geodesically incomplete. If
there are holes we have

\begin{theorem} \label{jse}
Every inextendible future holed spacetime admits a  future lightlike
 incomplete geodesic  and a future timelike  incomplete geodesic.
These geodesics are contained in $D(S)$ and the Riemann tensor, and
its covariant derivatives at any order, evaluated on a parallely
transported base over them have a finite limit.
\end{theorem}

\begin{proof}

%If the point $r$ there selected does not belong to
%$\overline{\varphi(H^{+}(S)}$ then $(M,g)$ is extendible and it is
%easy to find the required geodesics in $\varphi^{-1}(U)$. If we are
%in the other case in which

%The geodesic $\varphi^{-1}\circ \beta'$ gives the desired past
%inextendible lightlike incomplete geodesic.
We proceed as in the previous proof using the same notation. Any
causal future directed geodesic ending at $p$  in $N$ provides, once
pulled back to $M$ through $\varphi^{-1}$, the desired incomplete
geodesic. Since all these geodesics have image under $\varphi$ which
converge to $p$, the Riemann tensor and its covariant derivatives
must converge to the value they take on $p$.
\end{proof}

\subsection{Examples} \label{exa}
In this section we give some examples which clarify the independence
of the concept of hole and inextendibility, and the relative
strength of casual continuity, causal simplicity, and hole-freeness.

The dimensionality of the next examples is not really important
because one can always increase the dimension with the direct
product with $\mathbb{R}^k$, $k\ge 1$, with the Euclidean metric. It
is also useful to recall that any non-total imprisoning spacetime
can be made causally geodesically complete and hence inextendible
through multiplication of the metric by a suitable conformal
factor.\cite{beem76}

\begin{example} \label{pom}
{\em An extendible but hole-free spacetime}. The region $t<0$ of
Minkowski spacetime, or the region $t+x<0$ of Minkowski spacetime.
\end{example}

\begin{example} \label{miv}
{\em An inextendible causally continuous spacetime can be holed}.
Let $M$ be Minkowski 3+1 spacetime with the origin $o$ removed. The
spacetime metric is multiplied by a conformal factor different from
unity in $I^{+}(o)$ in such a way that the timelike geodesic $t>0,
x=y=z=0$, becomes geodesically complete in the past direction.
\end{example}

\begin{example} \label{pok}
{\em Another inextendible but holed spacetime}. The covering of 1+1
Minkowski spacetime minus a point.
\end{example}

\begin{example} \label{ouk}
{\em An inextendible causally continuous hole-free spacetime need
not be causally simple}. Plane wave gravitational metrics are known
to be causally continuous but not causally
simple.\cite{penrose65,ehrlich92} Furthermore, they are complete
\cite[Chap. 13]{beem96} thus inextendible \cite[Prop. 6.16]{beem96}
and by theorem \ref{jse} they are hole-free.
\end{example}

\begin{example} \label{our}
{\em An inextendible hole-free causal spacetime need not be causally
continuous}. Some generalized gravitational wave metrics are known
to be non-distinguishing.\cite{hubeny05}
\end{example}

\begin{example} \label{cga}
{\em Some holes can be removed extending the spacetime}. Minkowski
spacetime minus a point.
\end{example}

\begin{example}
{\em A future holed spacetime need not be past holed}. Minkowski 1+1
spacetime minus a future directed lightlike ray.
\end{example}

These examples prove that inextendibility and hole-freeness are
independent concepts (examples \ref{pom}, \ref{miv}, \ref{pok}).
Moreover, they prove that theorem \ref{nxp} is optimal since we
cannot weaken the assumption of causal simplicity to causal
continuity (example \ref{miv}). Finally, they show that the
hole-free property cannot find a place in the causal ladder between
causal simplicity and causal continuity (example \ref{our}), and
that in general hole-freeness does not promote causal properties
(example \ref{our}).

This fact had to be expected because, while the levels of the causal
ladder represent physical conditions on spacetime, the hole-free
property is a mathematical requirement which should be placed on the
very definition of spacetime and which cannot, by itself, improve
causality. Of course, our theorem shows that we could omit to
include the hole-free condition whenever dealing with globally
hyperbolic or causally simple spacetimes. Nevertheless, this
condition is expected to be useful for the study of causally
continuous and stably causal spacetimes, and in general for
spacetimes satisfying weak causality conditions.

\section{Conclusions}

We have introduced a definition of holed spacetime which removes
some undesirable consequences of previous proposals. In particular,
we have observed that a hole should generate an horizon $H(S)$ that
must in part disappear when we pass to a different spacetime
containing $\tilde{D}(S)$. We have then been able to prove that
every inextendible causally simple spacetime is hole-free, thus
confirming the expectation that causal simplicity removes holes from
spacetime. This theorem is optimal as it cannot be improved
weakening causal simplicity to causal continuity (example
\ref{miv}). Physically, this result means that if spacetime contains
some partial Cauchy hypersurface which does not completely develop
its influence, then there is some discontinuity in the causal
relation.

\section*{Acknowledgments}
I thank Alfonso Garc{\'\i}a-Parrado G\'omez-Lobo and the Centro de
Matematica, Universidade do Minho, for kind hospitality. This work
has been partially supported by GNFM of INDAM and by FQXi.

%\bibliography{../../bibliografie/simultaneity,../../bibliografie/libri,../../bibliografie/miei,../../bibliografie/mieiPreprints,../../bibliografie/mieiProceedings}
%\bibliographystyle{jmpTitles}

\end{document}